\documentclass[12pt,leqno]{amsart}

\usepackage{amsmath,amssymb,amscd,amsthm}
\usepackage{latexsym}
\usepackage{graphicx}

\newtheorem{theorem}{Theorem}
\newtheorem{lemma}{Lemma}
\newtheorem{proposition}{Proposition}

\newtheorem{definition}{Definition}
\newtheorem{corollary}{Corollary}

\newtheorem{claim}{Claim}

\newcommand{\f}[2]{\frac{#1}{#2}}
\newcommand{\dpr}[2]{\langle #1,#2 \rangle}

\newcommand{\al}{\alpha}

\newcommand{\de}{\delta}

\newcommand{\De}{\Delta}
\newcommand{\ve}{\varepsilon}

\newcommand{\ka}{\kappa}
\newcommand{\la}{\lambda}

\newcommand{\La}{\Lambda}
\newcommand{\si}{\sigma}

\newcommand{\Om}{\Omega}

\newcommand{\rn}{{\mathbf R}^n}

\newcommand{\rone}{\mathbf R^1}

\newcommand{\cm}{\mathcal M}

\newcommand{\cd}{\mathcal D}

\newcommand{\p}{\partial}

\newcommand{\beq}{\begin{equation}}
\newcommand{\eeq}{\end{equation}}
\newcommand{\beqna}{\begin{eqnarray*}}
\newcommand{\eeqna}{\end{eqnarray*}}
\newcommand{\beqn}{\begin{equation*}}
\newcommand{\eeqn}{\end{equation*}}
\newcommand{\bp}{\begin{proof}}
\newcommand{\ep}{\end{proof}}
\newcommand{\bprop}{\begin{proposition}}
\newcommand{\eprop}{\end{proposition}}
\newcommand{\bt}{\begin{theorem}}
\newcommand{\et}{\end{theorem}}
\newcommand{\bex}{\begin{Example}}
\newcommand{\eex}{\end{Example}}
\newcommand{\bc}{\begin{corollary}}
\newcommand{\ec}{\end{corollary}}
\newcommand{\bcl}{\begin{claim}}
\newcommand{\ecl}{\end{claim}}
\newcommand{\bl}{\begin{lemma}}
\newcommand{\el}{\end{lemma}}

\begin{document}

\title
[Solitary waves for Hertzian chains]
{On the existence of solitary traveling waves for generalized Hertzian chains}

\author{Atanas Stefanov}
\author{Panayotis Kevrekidis}

\address{Atanas Stefanov\\ 
Department of Mathematics \\
University of Kansas\\
1460 Jayhawk Blvd\\ Lawrence, KS 66045--7523}
\address{Panayotis Kevrekidis\\
Lederle Graduate Research Tower\\ 
Department of Mathematics and  Statistics\\
University of Massachusetts\\
Amherst, MA 01003}

\email{stefanov@math.ku.edu}
\email{kevrekid@math.umass.edu}

\thanks{Stefanov's research is supported in part by  
 NSF-DMS 0908802.  Kevrekidis is supported by NSF-DMS-0806762,
NSF-CMMI-1000337, as well as by the Alexander von Humboldt
Foundation and the Alexander S. Onassis Public Benefit Foundation
(RZG 003/2010-2011). Kevrekidis is grateful to Dr. G. Theocharis
for numerous discussions on this theme and to
Prof. Chiara Daraio for stimulating his interest on this subject.}
\date{\today}



\begin{abstract}
We consider the question of existence of 
``bell-shaped'' (i.e. non-increasing for $x>0$ and non-decreasing for $x<0$)  traveling waves for the strain variable
of the generalized Hertzian model describing, in the special
case of a $p=3/2$ exponent, the dynamics
of a granular chain. 
The proof  of existence of such waves is based on the English 
and Pego [Proceedings of the AMS {\bf 133}, 1763 (2005)]
formulation of the problem. More specifically, we construct an appropriate 
energy functional, for which we show that the constrained minimization 
problem over 
bell-shaped entries has a solution. We also provide an alternative proof of 
the Friesecke-Wattis 
result [Comm. Math. Phys {\bf 161}, 394 (1994)], by using the same 
approach (but where the 
minimization is not constrained over bell-shaped curves). 
We briefly discuss and illustrate numerically 
the implications on the doubly exponential decay
properties of the waves, as well as touch upon the modifications
of these properties in the presence of a finite precompression
force in the model.
\end{abstract}

\maketitle

\section{Introduction}

Localized modes on nonlinear lattices have been a topic
of wide theoretical and experimental investigation in a wide
range of areas over the
past two decades. This can be seen, e.g.,  in the recent general
review \cite{flagor}, as well as inferred from the topical reviews
in nonlinear optics \cite{leder}, atomic physics \cite{markus}
and biophysics \cite{peyrard} where relevant discussions have been
given of the theory and corresponding applications. 

One of the areas in which the theoretical analysis has been
especially successful in describing experimental data and
providing insights has been that of granular crystals
\cite{nesterenko1}.
These consist of closely-packed chains of elastically interacting
particles, typically according to the so-called Hertz contact law.
The broad interest in this area has emerged due to the
wealth of available material types/sizes (for which the
Hertzian interactions are applicable)
and the ability to tune the
 dynamic response of the crystals 
to encompass linear, weakly nonlinear, and strongly nonlinear
 regimes \cite{nesterenko1,sen08,nesterenko2,coste97}.  
This type of  flexibility renders these crystals perfect candidates
for many engeenering applications, including shock and energy absorbing 
layers \cite{dar06,hong05,fernando,doney06}, actuating devices
\cite{dev08},  and sound scramblers \cite{dar05,dar05b}. 
It should also be noted that another aspect of such systems
that is of particular appeal is their potential (and controllable)
heterogeneity which gives rise to the potential not only for modified
solitary wave excitations \cite{chiaraus}, but also for 
discrete breather ones \cite{chiarag}.

Another motivation for looking at waves in such lattices stems
from FPU type problems \cite{FPU0,FPU1}. In the prototypical FPU context,
it has been rigorously proved that traveling waves
exist which can be controllably approximated (in the appropriate
weakly supersonic limit) by solitary waves of the Korteweg-de Vries
equation \cite{pegogf1}. However, in more strongly nonlinear regimes,
compact-like excitations have been argued to exist
\cite{nesterenko1,sen08,coste97} (see also \cite{flach1} for
breather type excitations) and have even been computed 
numerically through iterative schemes \cite{pego,pikovsky},
but have not been rigorously proved to exist in the general case.
In the work of \cite{mackay}, the special Hertzian case
was adapted appropriately to fit the assumptions of 
the variational-methods' based proof of the traveling wave
existence theorem of \cite{fw} in order to establish these solutions.
However the proof 
does not give information on the wave profile.

Our aim herein is to provide a reformulation 
and illustration of existence of ``bell-shaped'' 
traveling waves in generalized Hertzian lattices. 
Our work is based on the iterative
schemes that have been previously presented in \cite{pikovsky,pego}
for the computation of traveling waves in 
such chains of the form:
\begin{eqnarray}
\ddot{v}_n = [v_{n-1}-v_n]_+^p - [v_n-v_{n+1}]_+^p.
\label{eqn1}
\end{eqnarray}
Here $v_n$ denotes the displacement of the n-th bead from its
equilibrium position.
The special case of Hertzian contacts is for $p=3/2$, but
we consider here the general case of nonlinear interactions with $p>1$.
Notice that the ``+'' subscript in the equations indicates that 
that the quantity in the bracket is only evaluated if positive,
while it is set to $0$, if negative (reflecting in the latter case
the absence of contact between the beads).
The construction of the traveling waves and the derivation
of their monotonicity properties will be based on the strain variant of the
equation for $u_n=v_{n-1}-v_n$ such that:
\begin{eqnarray}
\ddot{u}_n=[u_{n+1}]_+^p - 2 [u_n]_+^p + [u_{n-1}]_+^p,
\label{eqn2}
\end{eqnarray}

Our presentation will proceed as follows. In section 2, we will
give a preliminary mathematical formulation to the problem, 
briefly illustrate its numerical solution and some of its consequences.
Then, we will proceed in section 3 to state and prove our
main result. Some technical aspects of the problem will be
relegated to the appendices of section 4.

\section{Preliminaries and Numerical Results}

When seeking traveling wave solutions of the form
$u_n=u(x) \equiv u(n-c t)$, we are led to the advance-delay equation
(setting $c=1$)
\begin{equation}
\label{1}
u''(x)=u^p(x+1)-2 u^p(x)+u^p(x-1). \ \ x\in \rone 
\end{equation} 
where $u$ is a smooth and positive function, with the 
desired monotonicity involving decay in $(0,\infty)$ and 
increase in $(-\infty, 0)$. 
\subsection{Fourier transform and Sobolev spaces} 
We introduce the Fourier transform and its inverse via 
\begin{eqnarray*}
& & \hat{f}(\xi)= \int_{-\infty}^\infty f(x) e^{- 2\pi  i x \xi} dx \\
& & f(x)=    \int_{-\infty}^\infty \hat{f}(\xi) e^{ 2\pi i x \xi} d\xi
\end{eqnarray*}
  As is well-known, 
the second derivative operator $\p_x^2$ has a simple representation via the Fourier transform, namely 
$$
\widehat{\f{d^2}{dx^2} f}(\xi)=- 4\pi^2 \xi^2 \hat{f}(\xi). 
$$
For every $s\geq 0$, we may define 
$\widehat{(-\f{d^2}{dx^2})^s f}(\xi)=(4\pi^2 \xi^2)^s \hat{f}(\xi)$ and 
the Sobolev spaces $W^{s,p}$ via 
$$
\|f\|_{W^{s,p}}=\|f\|_{L^2}+ \|(-\De)^{s/2} f\|_{L^p}. 
$$
We will also consider the operator 
$$
\De_{disc} f(x)= f(x+1)-2 f(x)+f(x-1)
$$
on the space of $L^2(\rone)$ functions.
Using Fourier transform, we may write 
$$
\De_{disc} f(x)=  \int_{-\infty}^\infty \hat{f}(\xi)(e^{2\pi i \xi}+e^{- 2\pi i \xi}-2)
e^{ 2\pi i x \xi}d\xi = -4 \int_{-\infty}^\infty  \sin^2(\pi \xi)  \hat{f}(\xi)e^{2\pi  i x \xi}d\xi. 
$$
In other words, $\De_{disc}$ is given by the symbol $-4\sin^2(\pi \xi)$, that is 
\begin{equation}
\label{PE}
\widehat{\De_{disc} f}(\xi)=-4\sin^2(\pi \xi) \hat{f}(\xi). 
\end{equation}
\subsection{The English-Pego formulation} 
We may rewrite the equation \eqref{1} in the form 
\begin{equation}
\label{p:1}
u''(x)=\De_{disc}[u^p](x)
\end{equation}
Taking Fourier transform on both sides of this (and using   \eqref{PE}),
allows us to write \\  $-4\pi^2 \xi^2 \hat{u}(\xi)=-4 \sin^2(\pi \xi)\widehat{u^p}(\xi)$  or 
\begin{equation}
\label{4}
\hat{u}(\xi)=\f{\sin^2(\pi \xi)}{\pi^2 \xi^2} \widehat{u^p}(\xi)
\end{equation}
Equivalently, taking $\La: \hat{\La}(\xi)= \f{\sin^2(\pi \xi)}{\pi^2 \xi^2}$, 
\begin{eqnarray}
u(x)=  \La* u^p (x)=\int_{-\infty}^\infty \La(x-y) u^p(y) dy=:\cm[u^p].
\label{numerics}
\end{eqnarray}
In other words, we have introduced the convolution operator $\cm$ 
with kernel $\La(\cdot)$.
It is easy to compute that $\La (x)=(1-|x|)_+$ or 
$$
\La(x)=\left\{\begin{array}{l l} 
1-|x| & |x|\leq 1, \\
0 & |x|>1.
\end{array}\right.
$$
Note that  we have the following formula for the convolution $\La*f$ 
\begin{equation}
\label{2}
\cm f=\La*f(x)=\int_{x-1}^{x+1} (1-|x-y|) f(y) dy.
\end{equation}
\subsection{Numerical computations
 and other consequences of the English-Pego formulation}
For reasons of completeness and in order to appreciate the
form of (suitably normalized) solutions of Eq. (\ref{numerics}),
in Fig. \ref{fig1}, we used this equation as a numerical
scheme and proceed to iterate it until convergence. The
figure illustrates the converged profile $\phi$ of the
solution and its corresponding momentum $\phi_t=-c \phi_x$
(for $c=1$). The results of these
computations are shown for different values of $p$ (in order
to yield a sense of the $p-$dependence of the solution, namely 
for $p=3/2$
(the Hertzian case), $p=2$ and $p=3$ (the FPU-motivated cases,
in that they are the purely nonlinear analogs of $\alpha$-
and $\beta$-FPU respectively) and finally $p=10$ (as a 
large-$p$ case representative).
The figure shows the solutions' profile and corresponding
momenta, as well as the semi-logarithmic form of the profile,
so as to clearly illustrate the  doubly
exponential nature of the decay (see below). 
Notice that as $p$ increases, the decay
becomes increasingly steeper. 

\begin{figure}[tbp]
\includegraphics[width=6cm,height=6cm,angle=0,clip]{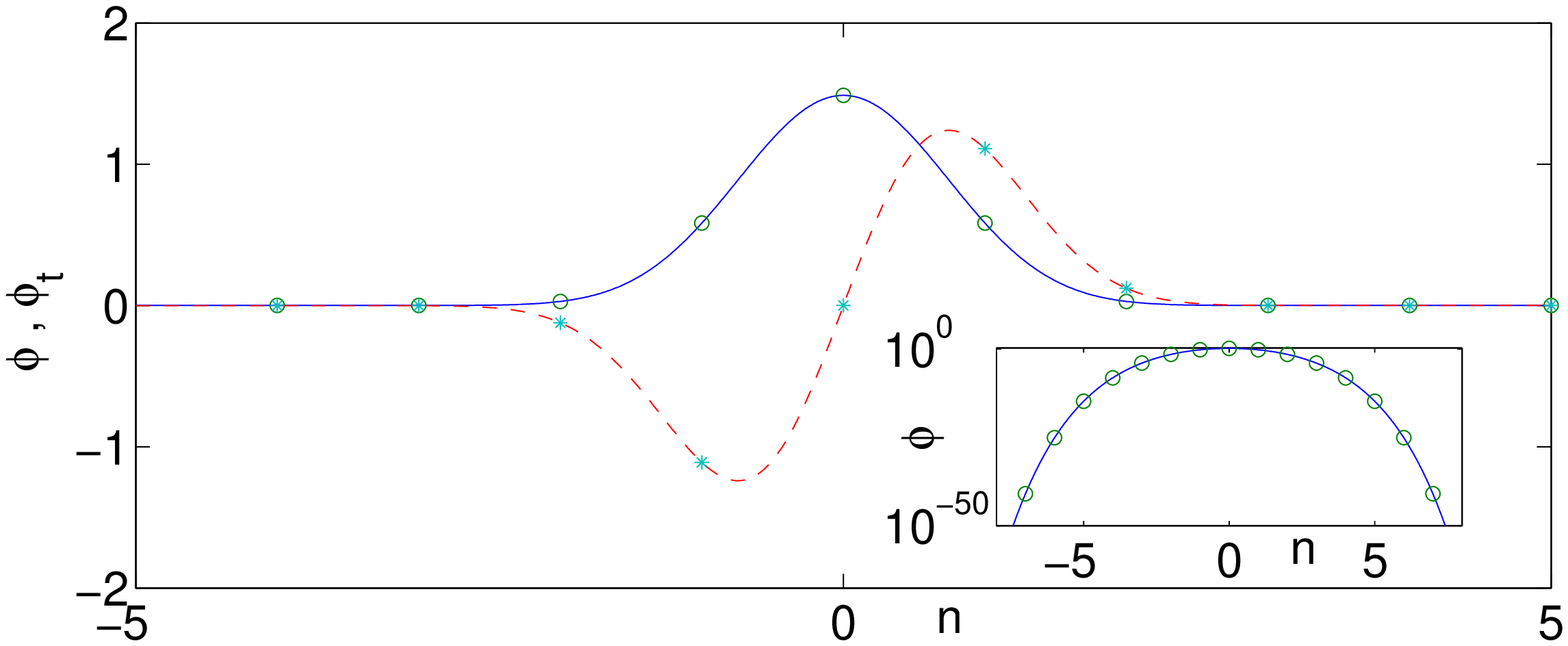}
\includegraphics[width=6cm,height=6cm,angle=0,clip]{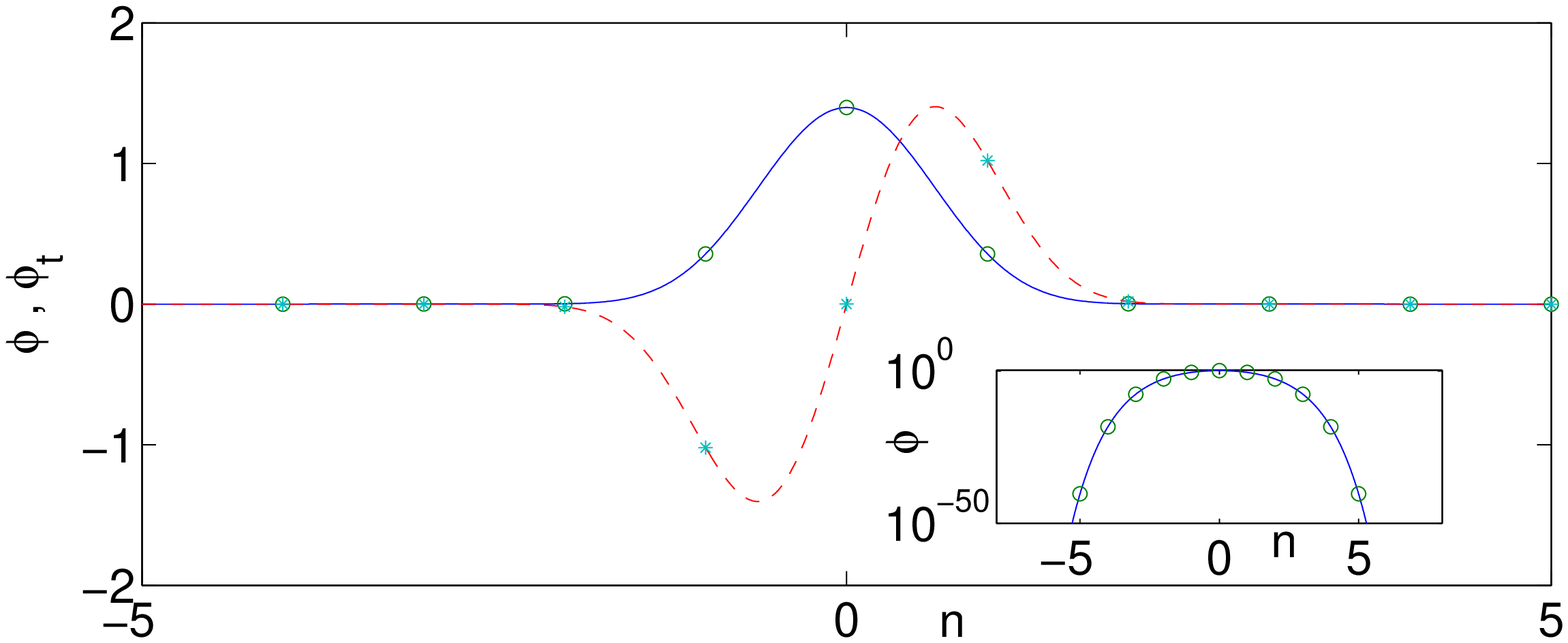}
\includegraphics[width=6cm,height=6cm,angle=0,clip]{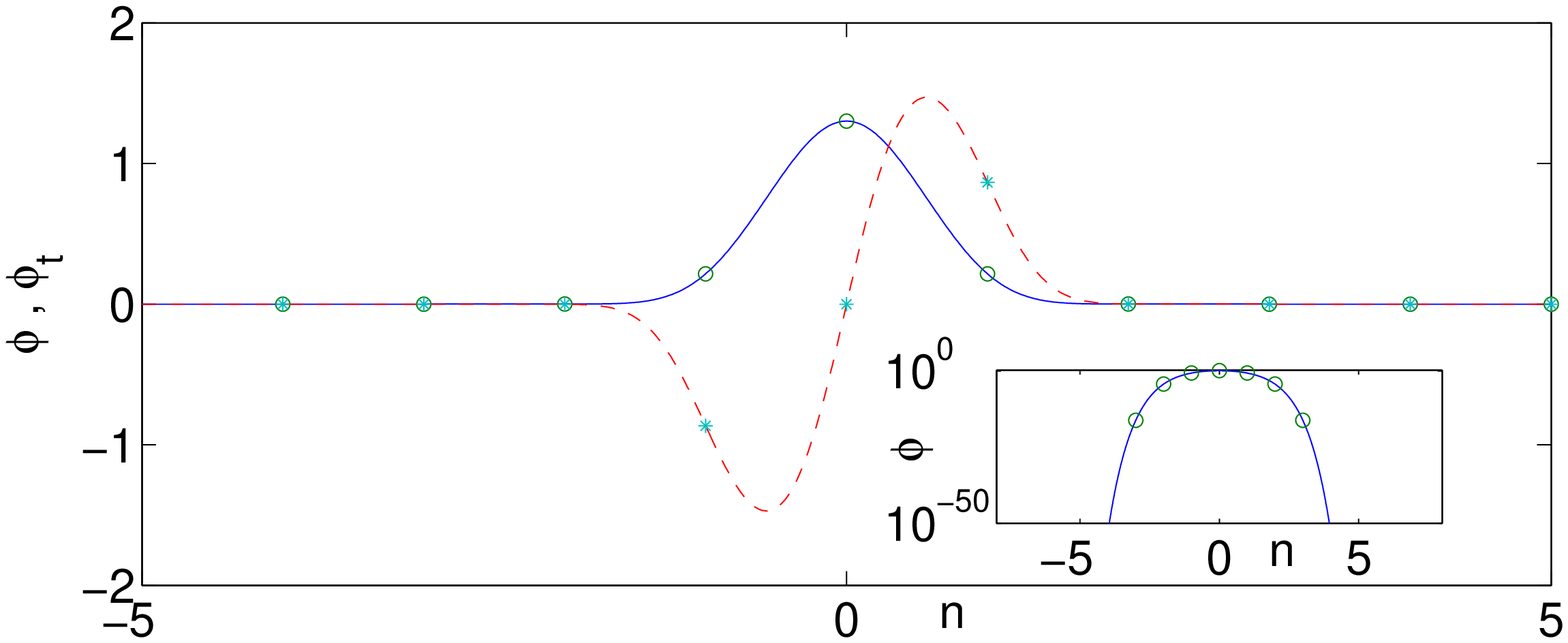}
\includegraphics[width=6cm,height=6cm,angle=0,clip]{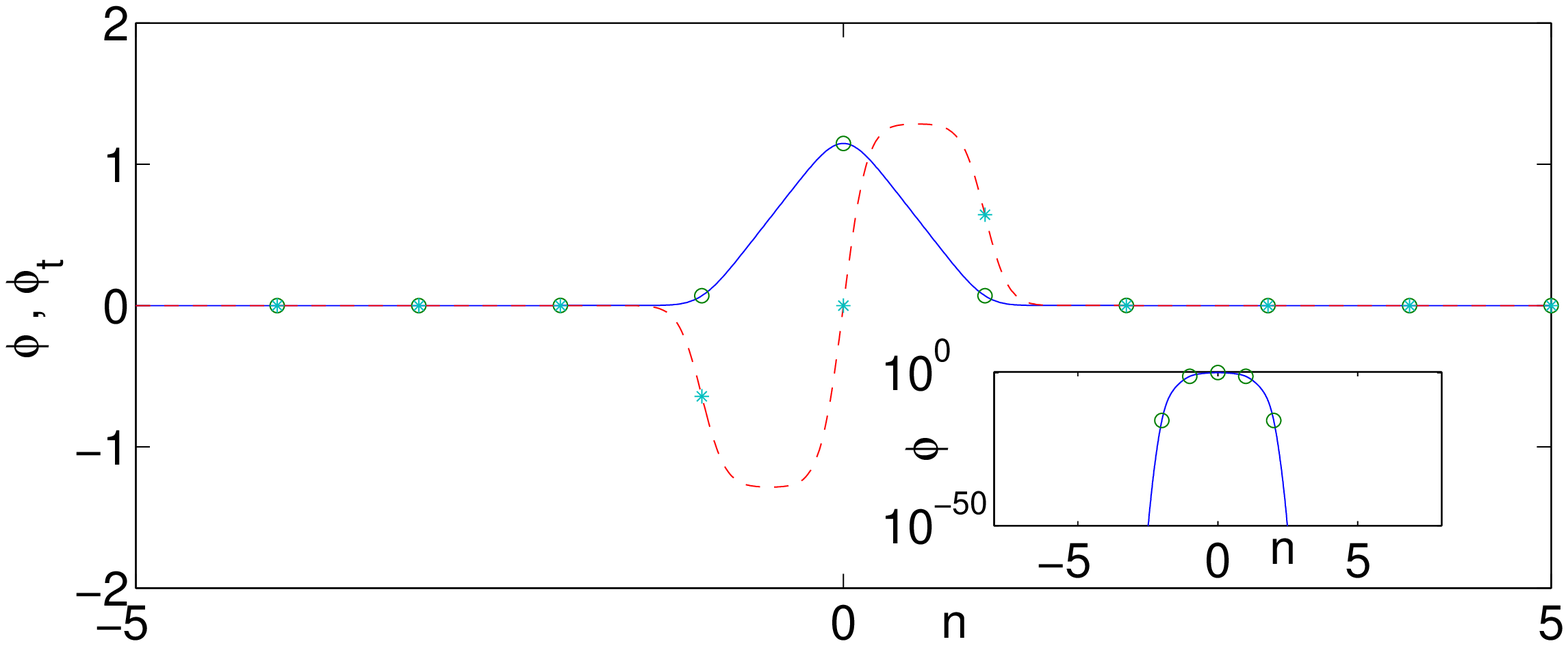}
\caption{Each one of the panels illustrates the numerically exact 
(up to a prescribed tolerance which we set here to $10^{-8}$) solution
profile of the iterative scheme, as renormalized for use
in Eq. (\ref{eqn2}). The solid (blue) line illustrates the
spatial form of the solution and the dashed (red) line the
corresponding momentum (for speed $c=1$). The circles and
stars denote respectively the ordinates of the lattice nodes
(extracted for use in Eq. (\ref{eqn2}). The inset illustrates
the profile in a semilog to highlight the doubly exponential
nature of the decay (notice also the steepening as $p$ increases).
The top left panel is for $p=3/2$, the top right for $p=2$,
the bottom left for $p=3$ and finally the bottom right for $p=10$.
} 
\label{fig1}
\end{figure}

To corroborate the exact nature of such traveling wave solutions, 
once the solution was obtained, then the ``lattice ordinates''
of both the solution and its time derivative were extracted
and inserted as initial conditions for the dynamical evolution
of Eq. (\ref{eqn2}). The results of the relevant time integration
(using an explicit fourth-order Runge-Kutta scheme) are shown
in Fig. \ref{fig3}. It can be straightforwardly observed that
excellent agreement is obtained with the  
 expectation of a genuinely traveling (without
radiation) solution with a speed of $c=1$, so that its center
of mass moves according to $x=c t$ (the solid line in the figure).
This confirms the usefulness of the method (independently of the
nonlinearity exponent $p$, as long as $p>1$) in producing 
accurate traveling solutions for this dynamical system.

\begin{figure}[tbp]
\includegraphics[width=6cm,height=6cm,angle=0,clip]{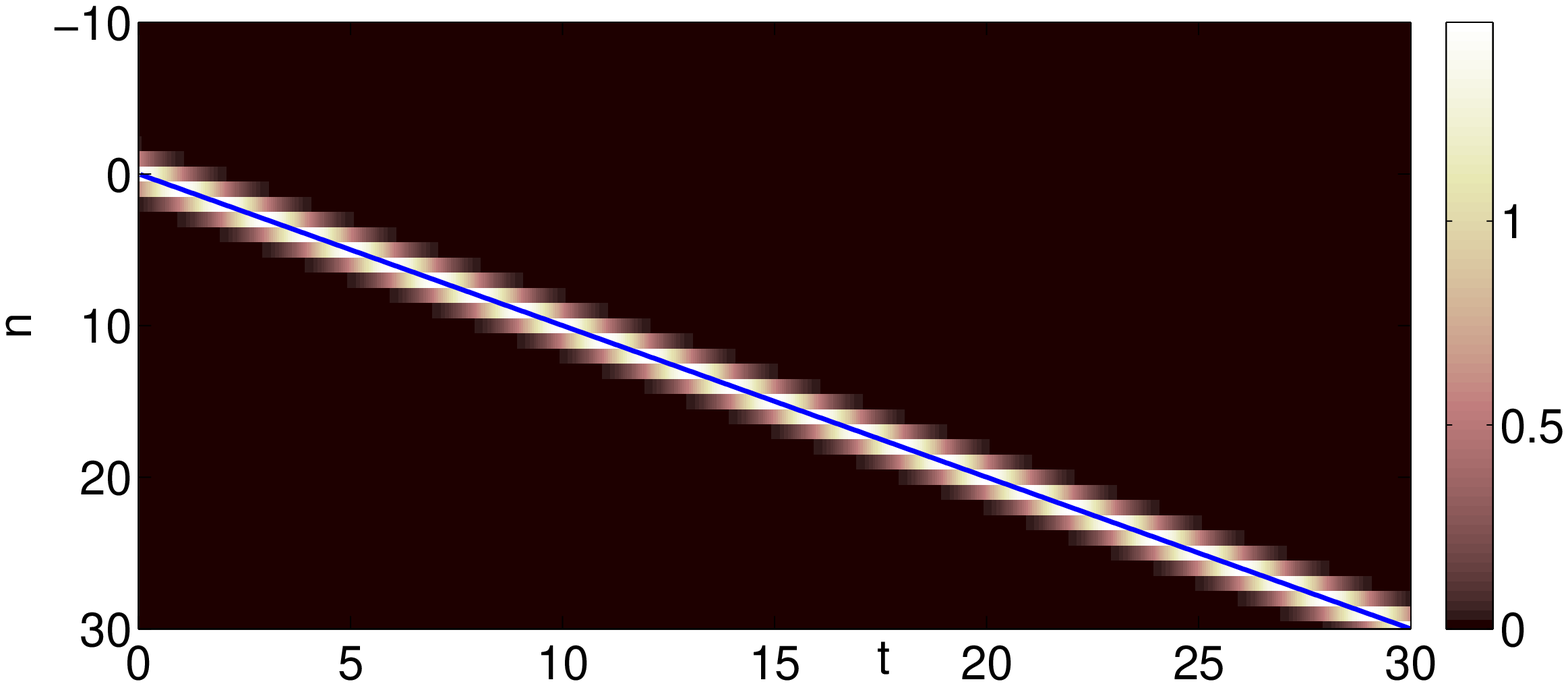}
\includegraphics[width=6cm,height=6cm,angle=0,clip]{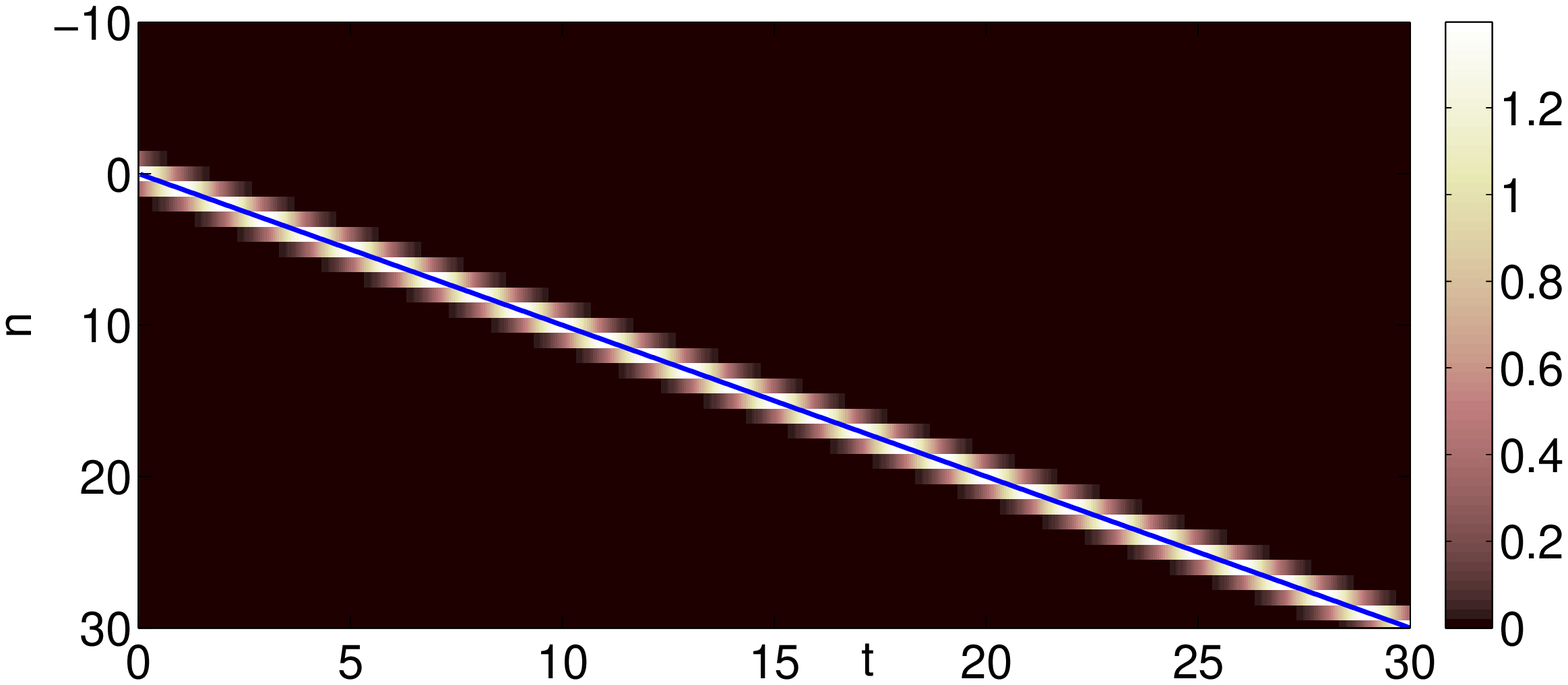}
\includegraphics[width=6cm,height=6cm,angle=0,clip]{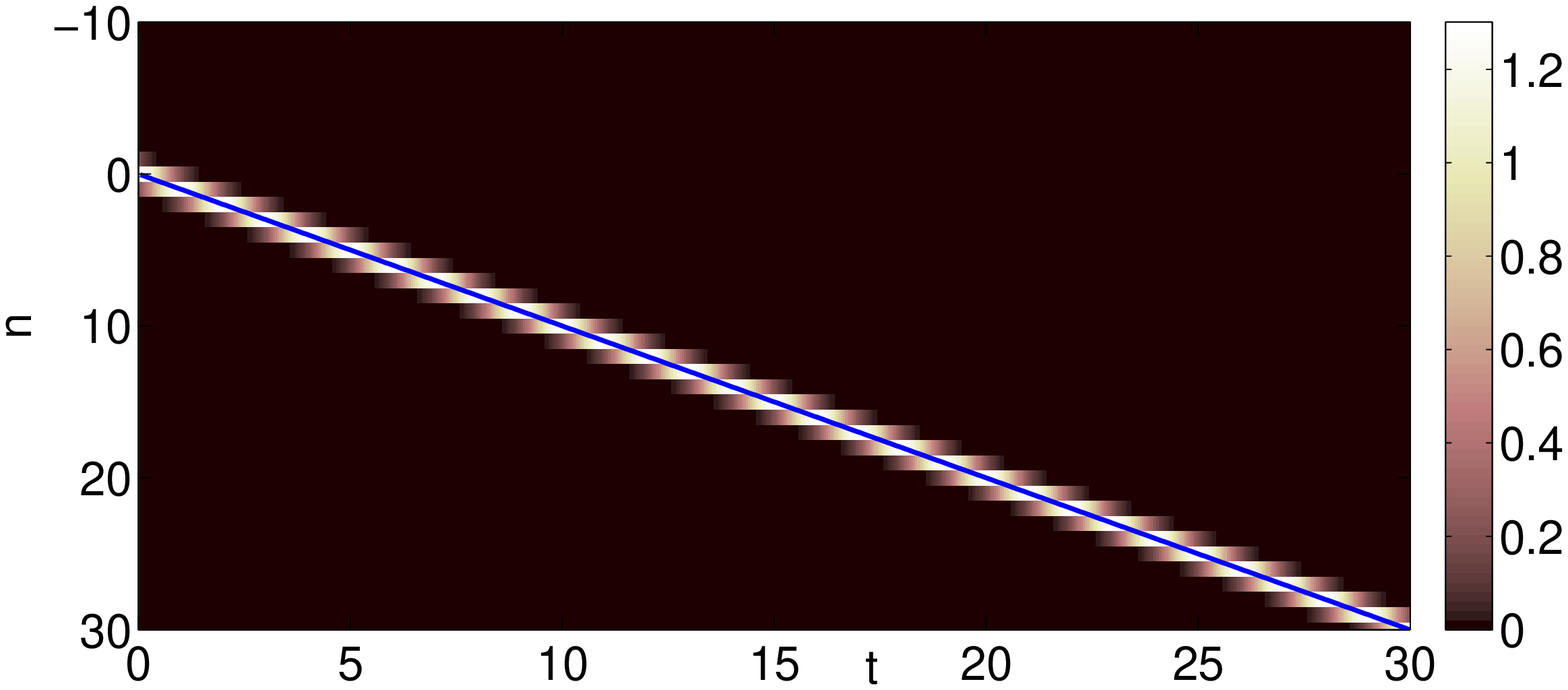}
\includegraphics[width=6cm,height=6cm,angle=0,clip]{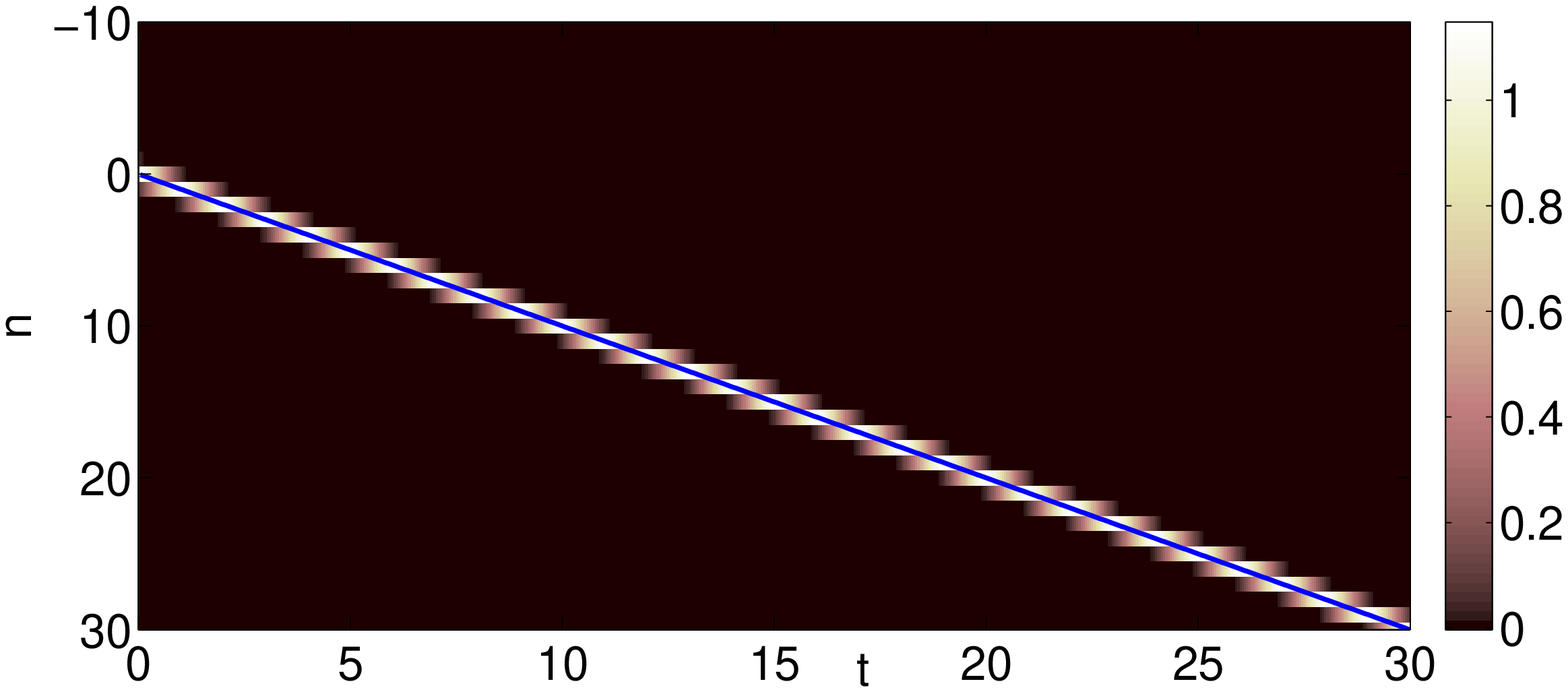}
\caption{Result of direct integration of Eq. (\ref{eqn2}),
with $r_n(0)$ and $\dot{r}_n(0)$ as seeded from the iteration
scheme's convergent profile. The solid line in each case illustrates
the trajectory of $x=c t$ (for $c=1$) to which the solutions
correspond.
One can notice for all values of $p$ ($p=3/2$: top left; $p=2$: top
right;
$p=3$ bottom left and $p=10$ bottom right) the agreement
with the expectation of a genuinely traveling (non-radiating)
waveform of $c=1$.} 
\label{fig3}
\end{figure}

If the convergence to such a nontrivial profile is established
(as we will establish it in section 3 with the proper monotonicity
properties based on our modified variational formulation),
there is an important immediate conclusion about the decay
properties of such a profile. In particular, 
\begin{eqnarray}
u(x+1) =\int_{-1}^1 \Lambda(y) u^p(x+1-y) dy 
\leq 
u^p(x) \Rightarrow u(x+n) \leq u(x)^{p^n},
\label{eqn23}
\end{eqnarray}
Hence, as was originally discussed in \cite{chatter} and then
more rigorously considered in \cite{pego} (see also \cite{pikovsky}),
the solutions of Fig. \ref{fig1} feature a doubly exponential decay.
This very fast decay (and nearly compact shape) of the pulses can be clearly
discerned in the semi-logarithmic plots of the figure.

As a slight aside to the present considerations, we should mention
that a physically relevant variant of the problem consists of
the presence of a finite precompression force $F_0$ at
the end of the chain \cite{nesterenko1,sen08}. In that case, the model
of interest becomes (in the strain formulation and with $F_0=\delta_0^p$)
\begin{eqnarray}
\ddot{u}_n=
[\delta_0 + u_{n+1}]_+^p - 2 [\delta_0 + u_n]_+^p + [\delta_0 + u_{n-1}]_+^p,
\label{eqn2_mod} 
\end{eqnarray}
The case of $\delta_0=0$ constitutes the so-called sonic vacuum 
\cite{nesterenko1}, while that of finite $\delta_0$ features a finite
speed of sound (and allows the existence and propagation of linear spectrum
excitations). It is worthy then to notice that for $\delta_0 \neq 0$,
the above decay estimate is modified as:
\begin{eqnarray}
u(x+1) \sim \delta_0^{p-1} u(x) \Rightarrow u(x) \sim 
\delta_0^{p-1} \exp(n \log{r(x_0)})
\label{precomp}
\end{eqnarray}
Namely, the solutions are no longer doubly exponentially localized but
rather feature an exponential tail (and are progressively closer to
regular solitary waves). This can be thought of as a ``compacton
to soliton'' transition that is worth exploring further (although the
case of $\delta_0 \neq 0$ will not be considered further herein).

\subsection{Some facts and definitions regarding distributions} 
We now turn to several definitions, which will be useful in the sequel. Our space of test functions will be the following. For  $V\subset \rone$ -  an open set, let  
 $\cd(V)=C^\infty_0(V)$ be the set of all $C^\infty$ functions with compact support, contained inside $V$. We equip this with the usual topology of a Frechet space,  generated by a family of seminorms $p_N(f)=\sup_{x\in V_n} \sum_{\al=1}^N |\p_x^\al f(x)|$, where $\{V_n\}_n$ is some fixed nested family of compact sets, so that $\cup_n V_n=V$. 
 The distributions over this space of functions, which we denote by $\cd'(V)$, is its dual space, namely all continuous linear functionals over $\cd(V)$. The derivatives of such distributions are defined in the usual way $\dpr{h'}{\psi}:=-\dpr{h}{\psi'}$. One may also define a convolution of a distribution $h\in \cd(\rone)$ 
 with a given $C^\infty$ function  $\chi$, by $\dpr{h*\chi}{\psi}:=\dpr{h}{\psi*\chi(-\cdot)}$. 
 
 We say that two distributions are in the relation $h_1\leq h_2$ in $\cd'(V)$ sense, if for all $\psi\in \cd(V), \psi\geq 0$, we have $\dpr{h_1}{\psi}\leq \dpr{h_2}{\psi}$. In particular, 
 \begin{definition}
 \label{defi:1}
 We say that a distribution $h$ is non-increasing (non-decreasing)
  over a set $V$, if $h'\leq 0$ ($h'\geq 0$)  in $\cd'(V)$ sense. 
 \end{definition}
 Of course, if $h$ happens to have a locally integrable derivative $h'$ on an interval, 
 then the notion of non-increasing function coincides with the standard (pointwise) notion by the fundamental theorem of calculus. More generally, we have the following 
 \begin{lemma}
 \label{le:p}
 Suppose that $h$ is a locally integrable function in $(a,b)$ and it satisfies $h'\geq 0$ ($h'\leq 0$) in $\cd'(a,b)$ sense. Then, $h$ is almost
everywhere (a.e.) non-decreasing (non-increasing, respectively) function on $(a,b)$. That is, for almost all pairs $x<y$, $h(x)\leq h(y)$ ($h(x)\geq h(y)$ respectively).  
 \end{lemma}
 \begin{proof}
 It is well-known by the Lebesgue differentiation theorem 
 that for a locally integrable function $h$, one has $h(x)=\lim_{\de\to 0} \f{1}{\de}\int_{x}^{x+\de} h(y) dy$ $a.e.$ All such points $x$ are called Lebesgue points for $h$. Denote this full measure set by $L$. We will show that for all $a<c_1<c_2<b: 
 c_1, c_2\in L$, we have $h(c_1)\leq h(c_2)$. Indeed, let $\de>0$ be so that 
 $\de<\min(c_2-c_1, c_1-a, b-c_2)$. Define a function $\psi$, 
 $$
 \psi=\left\{ \begin{array}{l l} 
 0 & x<c_1-\de\\
 x-(c_1-\de) & c_1-\de<x\leq c_1 \\
 \de & c_1<x\leq c_2 \\
 c_2+\de-x & c_2<x\leq c_2+\de \\
 0 & c_2+\de\leq x.
\end{array} 
 \right. 
 $$
 Clearly $\psi\geq 0$ is not smooth, but is continuous and it can be approximated well by test functions. Moreover $\psi'=1$ on $(c_1-\de, c_1)$, $\psi'(x)=-1$ on 
 $(c_2, c_2+\de)$ and zero otherwise. Since $0\leq \dpr{h'}{\psi}=-\dpr{h}{\psi'}$ (and these are well-defined quantities), we obtain 
 $$
 \dpr{h}{\psi'}=\int_{c_1-\de}^{c_1} h(y) dy - \int_{c_2}^{c_2+\de} h(y) dy\leq 0. 
 $$
 This is true for all $\de>0$, which are sufficiently small. Thus, dividing by $\de>0$ and taking limit as $\de\to 0+$ (and taking into account that both $c_1, c_2$ are Lebesgue points), we conclude that $h(c_1)\leq h(c_2)$. 
 \end{proof}
 We find the following trick useful, which allows us to  reduce   
 non-increasing/non-decreasing distrubutions to non-increasing/non-decreasing functions. 
 
More precisely, let us fix a positive even  $C_0^\infty$ function $\Phi$, so that $supp\  \Phi\subset (-1,1)$ and 
  $\int \Phi(x) dx=1$. Let $\Phi_\de(x):=\f{1}{\de}\Phi(\f{x}{\de})$ and define $h_\de=h*\Phi_\de$. The following lemma has a  standard proof. 
  \begin{lemma}
  \label{le:kl}
  Let $h$ be a non-increasing (non-decreasing) distribution. Then for every $\de>0$, 
  $h_\de$ is a $C^\infty$ function, which is non-increasing (non-decreasing respectively). Moreover $h_\de\to h$ in the sense of distributions. 
  \end{lemma}
 \noindent  Next 
  \begin{definition}
  \label{defi:4} 
We say that a \underline{distribution $h\in \cd'(\rone)$ is  bell-shaped}, if $h$ is non-decreasing in  $\cd'(-\infty,0)$ and $h$ is non-increasing in $\cd'(0,\infty)$. 
  \end{definition}
In the sequel, we   need the following technical result. 
\begin{lemma}
\label{le:3}
Suppose that $f$ is an even distribution, so that $f$ is non-increasing in 
$\cd'(0,\infty)$ and non-decreasing in $\cd'(-\infty,0)$. Then   $\La*f$ is non-increasing in $(0,\infty)$ and non-decreasing in $(-\infty,0)$. 
 
Assume that $g$ is non-increasing in $\cd'(a,\infty)$ for some $a\in\rone$. Then $\La*g(x)$ is non-increasing in 
$(a+1,\infty)$. 

Assume that $h$ is non-decreasing in $(-\infty,a)$. Then, $\La*h$ is non-decreasing in 
$(-\infty,a-1)$. 
\end{lemma}
\begin{proof}
By Lemma \ref{le:kl}, it suffices to consider functions instead of ditributions with the said properties. 
We have the following computation for the derivative of the function $\La*g$
\begin{equation}
\label{m1}
(\La*g)'(x)=\int_{x}^{x+1} g(y) dy- \int_{x-1}^x g(y) dy
\end{equation}
which follows by differentiating \eqref{2}. 

It is immediate from \eqref{m1} that the claims for the functions $\La*g$ and $\La*h$ hold true.
 Regarding $\La*f$, it is clear that $\La*f$ is non-increasing in $(1,\infty)$ and non-decreasing in $(-\infty,-1)$. 
 Thus, we need to show that for  $x\in (0,1)$, $(\La*g)'(x)\leq 0$ and for $x\in (-1,0)$, $(\La*g)'(x)\geq 0$. We only verify this for $x\in (0,1)$, since the other inequality follows in a similar manner. Indeed, using $f(x)=f(-x)$, 
 $$
 (\La*f)'(x)=\int_{x}^{x+1} f(y) dy- \int_{x-1}^x f(y) dy=\int_{x}^{x+1} f(y) dy- (\int_0^x f(y) dy+\int_0^{1-x} f(y) dy).
 $$
 Since $f$ is non-increasing in $(0,\infty)$, 
 \begin{eqnarray*}
 & & 
 \int_0^x f(y) dy\geq \int_0^x f(y+x)dy=\int_x^{2x} f(y) dy, \\
 & & \int_0^{1-x} f(y)dy\geq \int_0^{1-x} f(y+2x)dy=\int_{2x}^{x+1} f(y) dy
 \end{eqnarray*}
 Going back to the expression for $(\La*f)'(x)$, this implies $(\La*f)'(x)\leq 0$, which was the claim. 
\end{proof}
 We will also need the following multiplier in our considerations 
 $$
 \widehat{Q f}(\xi)=\f{\sin(\pi \xi)}{\pi \xi} \hat{f}(\xi).
 $$
 It is actually easy to see that since $ \widehat{\chi_{[-\f{1}{2},\f{1}{2}]}}(\xi)= \f{\sin(\pi \xi)}{\pi \xi}$, we have also the representation 
 \begin{equation}
 \label{p:3}
 Q f(x)=\int_{x-1/2}^{x+1/2} f(y) dy. 
 \end{equation} 
 Noted that by our definition of the operator $\cm$, we have $\cm=Q^2$.

\section{Statement and Proof of the Main Result}

\begin{theorem}
\label{theo:1}
The equation \eqref{1} has a positive solution $u$, so that 
\begin{itemize} 
\item $u$ is even, 
\item $u$ is bell-shaped, 
\item $u\in H^\infty(\rone)$ 
\end{itemize} 

\end{theorem}

Notice that the fundamental contribution of our work to the setting
of generalized Hertzian lattices is the 
characterization of the (monotonically decaying on each side)
nature of the traveling waves.
We first explain the possible approaches to this problem. One approach is to use the form \eqref{2}, to show that the map $\phi\to \La* \phi^p$ has a fixed point. We shall pursue a different route, more in line with the work of
Friesecke and Wattis \cite{fw}.
Namely, we shall consider a constrained minimization problem associated   with \eqref{1}. We should note here, that the existence result of Friesecke-Wattis 
is based on the equivalent via the change of variables $u(x)=z(x-1)-z(x)$
(as discussed in section 1)  formulation 
\begin{equation}
\label{5}
z''(x)=[z(x-1)-z(x)]_+^p-[z(x)-z(x+1)]_+^p.
\end{equation}

We will in fact consider a different representation of \eqref{p:1}. To  that end, 
introduce a positive function $w: w^{1/p}=u$, whence \eqref{p:1} reduces\footnote{Note that this is a good transformation, since we are interested in positive solutions of \eqref{p:1}.}   to 
$$
(w^{1/p})''=\De_{disc}(w). 
$$
This is easily seen to be equivalent to $\widehat{w^{1/p}}(\xi)=\f{\sin^2(\pi \xi)}{\pi^2\xi^2} \hat{w}(\xi)$. Undoing the Fourier transform yields 
\begin{equation}
\label{p:2}
w^{1/p}=\cm w= Q^2 w= \La * w.
\end{equation}
Thus, we need to find a solution $w$ to \eqref{p:2}, which is as stated in Theorem \ref{theo:1}. 

Let $\ve>0$  be chosen appropriately small momentarily. 
Let $q=1+\f{1}{p}\in (1,2)$ and consider the following constrained optimization problem 
\begin{equation}
\label{7}
\left| 
\begin{array}{c}
J_\ve(v)=\int_{-\ve^{-1}}^{\ve^{-1}}|Q v(x)|^2 dx   \to \max \\
\textup{subject to} \ \  I(v)=\int_{\rone} v^q(x)dx =1,  \\ 
v\geq 0,  v - \textup{bell-shaped} 
\end{array}
\right.
\end{equation}
Solving \eqref{7} will eventually lead us to a solution of \eqref{p:2}. 
\subsection{Constructing a maximizer for \eqref{7}} 
\label{sec:3.1}
Let us first show that the expression $J_\ve(v)$ is bounded from above, if $v$ satisfies the constraints. Indeed, we have by Sobolev embedding 
\begin{eqnarray*}
J_\ve(v)\leq \|Q v\|_{L^2}^2\leq C_q \|Q v\|_{\dot{W}^{1/q-1/2,q}}^2. 
\end{eqnarray*}
By Gagliardo-Nirenberg inequality, we further bound 
$$ 
\|Q v\|_{\dot{W}^{1/q-1/2,q}}\leq \|Q v\|_{\dot{W}^{1,q}}^{1/q-1/2}\|Q v\|_{L^q}^{3/2-1/q}. 
$$
But by the definition of $Q$ (see \eqref{p:3}), we have 
\begin{eqnarray*}
& &  \|Q v\|_{\dot{W}^{1,q}}= \|\p_x[Q v]\|_{L^q}=
\|v(\cdot+1/2)-v(\cdot-1/2)\|_{L^q}  \leq 2 \|v\|_{L^q}=2,  \\
& & \|Q v\|_{L^q}^q\leq \int_{-\infty}^\infty (\int_{x-1/2}^{x+1/2} v(y) dy)^q dx\leq \|v\|_{L^q}^q=1. 
\end{eqnarray*}
Thus, $J_\ve(v)\leq 2 C_q$ and hence $J^{\max}_\ve =\sup_{v} J(v)\leq 2 C_q$, where $C_q$ is the square of the Sobolev embedding constant 
$\dot{W}^{1/q-1/2,q}(\rone)\hookrightarrow  L^2(\rone) $. 

We will now select $\ve_0$ so small that  $J^{\max}_\ve =\sup_{v} J(v)$ (which we have just shown exists) is positive for all $0<\ve<\ve_0$. Indeed, take $v_0(x)=c_q e^{-x^2}$, so that $c_q^q\int_{-\infty}^\infty e^{-q x^2} dx=1$. Thus, the function $v_0$ satisfies the constraints and therefore 
$$
J^{\max}_\ve\geq J(v_0)= c_q^2 \int_{-\ve^{-1}}^{\ve^{-1}} |\int_{x-1/2}^{x+1/2}
e^{-2 y^2} dy|^2 dx.
 $$
 Clearly, as $\ve\to 0$, we have that the right hand side converges to \\  
 $ c_q^2 \int_{-\infty}^{\infty} |\int_{x-1/2}^{x+1/2}
e^{-2 y^2} dy|^2 dx>0$. Thus, there exists $\ve_0$, so that for all $\ve\in (0,\ve_0)$, $J^{\max}_\ve>0$. In fact, we can select the $\ve_0$, so that 
 \begin{equation}
\label{p:55}
J^{\max}_\ve\geq \f{c_q^2}{2} \int_{-\infty}^{\infty} 
|\int_{x-1/2}^{x+1/2} e^{-2 y^2} dy|^2 dx,
 \end{equation}
 whenever $\ve<\ve_0$. 
For the rest of this section, fix $\ve<\ve_0$. Construct a maximizing sequence $v^n\in L^q(\rone):\|v^n\|_{L^q}=1$, so that $J_\ve(v^n)\to J^{\max}_\ve$. 
Namely, we take $v^n$ satisfying the constraints, so that $J_\ve(v^n)>J^{\max}_\ve-\f{1}{n}$. We only consider $n$ large enough, so that $\f{1}{n}<J^{\max}_\ve$, in which case $J_\ve(v^n)>0$.  

By the compactness of the unit ball of $L^q$ in the weak topology, 
we may take a weak $L^q$ limit $v=\lim_n v_n$.  Clearly, weak limits preserve 
the property that $v$ is even and that $v$  is a bell-shaped function.  We now need to show that $v$ satisfies the constraint $\|v\|_{L^q}=1$, which is non-trivial since norms are in general only 
lower semicontinuous  with respect to weak limits (and hence, we can only guarantee 
$\|v\|_{L^q}\leq 1$). 

We show now that, there exists a subsequence $\{n_k\}_k$ so that
\begin{equation}
\label{p:70}
\lim_k \int_{-\ve^{-1}}^{\ve^{-1}} |Q v^{n_k}(x)|^2 dx= \int_{-\ve^{-1}}^{\ve^{-1}} |Q v(x)|^2 dx
\end{equation}
Indeed, we check that 
\begin{eqnarray*}
\|Q v^n\|_{W^{1,q}} &=&\|\p_x[Q v^n]\|_{L^q}+\|Q v^n\|_{L^q}= 
\|v^n(\cdot+1/2)-v^n(\cdot-1/2)\|_{L^q}+ \|Q v^n\|_{L^q}\\
&\leq & C\|v^n\|_{L^q}=C
\end{eqnarray*}
Thus, by the compactness of the embedding $W^{1,q}(\rone) \Subset L^2(-\ve^{-1}, \ve^{-1})$ 
  we conclude that there is a subsequence $\{n_k\}_k$ and $z\in L^2$,  so that 
  $\|Q v^{n_k}-z\|_{L^2}\to 0$. By uniqueness of weak limits, (note that $Q v^{n_k} \rightharpoonup Q v$), it follows that  $z=Q v$ and hence \eqref{p:70}. 
  
Clearly now  
\begin{eqnarray*}
J^{\max}_\ve &=& \limsup_k J_\ve(v^{n_k})=\limsup_k 
(\int_{-\ve^{-1}}^{\ve^{-1}} |Q v^{n_k}(x)|^2 dx)= 
(\int_{-\ve^{-1}}^{\ve^{-1}} |Q v(x)|^2 dx) = J_\ve(v). 
\end{eqnarray*}
On the other hand, $\|v\|_{L^q}\leq 1$, by the lower semicontinuity of $\|\cdot\|_{L^q}$ with respect to weak limits, but clearly $v\neq 0$, since $J_\ve(v)= J_\ve^{\max}>0$.  
We will show that in fact $\|v\|_{L^q}=1$.  Indeed, assume the opposite $0<\rho=\|v\|_{L^q}<1$ and consider the function $v/\rho:\|v/\rho\|_{L^q}=1$.  Observe that 
$$
J_\ve(\f{v}{\rho})=J_\ve (v)\rho^{-2}= J^{\max}_\ve \rho^{-2}>J^{\max}_\ve.
$$
Thus, $\|v\|_{L^q}=1$ (otherwise, we get a contradiction with the constrained maximization problem \eqref{7}). This implies that $J_\ve(v)=J^{\max}_\ve$, otherwise, we   get a contradiction with the definition of $J^{\max}_\ve$. Thus, we have shown that the weak limit $v$ is indeed a maximizer for  \eqref{7}. 

\subsection{Euler-Lagrange equations for the maximizer of \eqref{7}}
\label{sec:3.2} 
Before we proceed to the actual proof, it is relevant
to mention a few words about our strategy. We will have no essential difficulties in deriving the Euler-Lagrange equation (see \eqref{p:15} below) on the set\footnote{for the precise definition of $\Om$ see below}  $\Om\sim \{x:v'(x)<0\}$ using the standard calculus of variations arguments. The reason  is  that on the compact subsets of  $\Om$, the maximizer $v$ is already a strictly decreasing function and therefore, for each fixed $z$, with support inside $\Om$ and regardless of its increasing/decreasing behavior, there will be $\la=\la(z)>0$, so that $v+\la z$ will be an acceptable function for the maximization problem \eqref{7}.  

On the other hand, we will have issues deriving appropriate equations on  $\rone\setminus \Om$, the reason being that on every non-trivial interval of the set, we will have 
$v=const.$ and hence  the perturbations $z$ must be increasing in this interval, in order for $v+\la z, \la>0$ to be an acceptable function for the maximization problem \eqref{7}. 

We will show that   
the set $\rone\setminus \Om$ consists of isolated points. To that end, under the 
assumption that there are non-trivial intervals inside $\rone\setminus \Om$, we 
derive an Euler-Lagrange equation for $v$ on such intervals (see \eqref{p:67} below), which in turn will imply that the maximizer $v$ is trivial, a contradiction.

First, we start with a technical issue. Since $v\geq 0$ and $v\downarrow$ in $(0,\infty)$, there is $\si=\inf\{x>0: v(x)=0\}\leq \infty$. If $\si<\infty$, we henceforth restrict our attention to $(-\si,\si)$. Clearly $v(y)=0$ for all $y\geq \si$.  

Consider perturbations of $v$ of the form $v+\la z$, where $\la>0$ and $z\in C^\infty_0(\rone)$, so that $\textup{supp}\  z\subset (-\ve^{-1}+1, \ve^{-1}-1) \cap (-\si,\si)$. We have 
\begin{eqnarray*}
 \int_{-\infty}^\infty (v(x)+\la z(x))^q dx & = & \int_{-\infty}^\infty v^q(x) dx +\la 
q\int_{-\infty}^\infty v^{q-1}(x) z(x) dx+ O(\la^2), \\
&=& 1+\la 
q\int_{-\infty}^\infty v^{q-1}(x) z(x) dx+ O(\la^2),  \\ 
 J_\ve(v+\la z) &=& \dpr{Q(v+\la z)}{Q(v+\la z)}_{L^2(-\ve^{-1}, \ve^{-1})}= \\
 &=& J_\ve(v)+
 2\la \dpr{Q^2 v}{z}  +O(\la^2).
\end{eqnarray*}
We now define the set $\Om$. Roughly speaking,  we would like to define \\ 
$\Om=\{x: v'(x)<0\}$. This is however impossible, since $v'$ is merely a distribution. 

Instead, we define $\Om^+$ to be the maximal open subset of $(0,\infty)$, 
so that for every compact subinterval of 
it $[c,d]\subset \Om^+$, there is $\de_{c,d}>0$ with 
$v'(x)\leq -\de_{c,d}$ in $\cd'(c,d)$ sense. Equivalently, we may define  
$$
\Om^+=\{x_0>0: \exists r_0, \de>0,  
\dpr{v'}{ \psi} \leq -\de \int\psi: \forall \psi\geq 0 \in \cd(x_0-r_0,x_0+r_0)\}
$$
Define $\Om^-=\{-x: x\in \Om^+\}$ and finally $\Om=\Om^-\cup \Om^+$. 
\subsubsection{Euler-Lagrange on $\Om$} 
\label{sec:3.2.1}
Due to our requirement for bell-shaped test functions in \eqref{7}, we need to impose extra restrictions on the function $z$. To that end, note that $\Om$ is an open set and 
 fix an interval  $[a,b] \subset \Om$. By the definition of $\Om$ and 
 compactness, it is clear that there exists $\de=\de_{a,b}$, so that $v'\leq -\de$ in $\cd'(a,b)$ sense. 
 
Fix a function $z\in C^\infty_0(a,b)$. Clearly, for each such $z$, there exists $\la_0=\la_0(z,\de_{a,b})$, so that for all $0<\la<\la_0$,  $(v+\la z)/\|v+\la z\|_{L^q}$ satisfy all constraints. That is, we do not need to restrict over $z$ in terms of positivity or bell-shapedness\footnote{which will be the case in the Section \ref{sec:4.2.2} below}.  
Thus, we have that for all $0<\la<\la_0(z)$, 
\begin{eqnarray*}
J_\ve\left(\f{v+\la z}{\|v+\la z\|_{L^q}}\right) &=&  \f{J_\ve(v+\la z)}{\|v+\la z\|_{L^q}^2}=\f{J^{\max}_\ve+
 2\la \dpr{\cm v}{z} +O(\la^2)}{(1+\la 
q\int_{-\infty}^\infty v^{q-1}(x) z(x) dx+ O(\la^2))^{2/q}}=\\
&=& \f{J^{\max}_\ve+
 2\la \dpr{\cm v}{z} +O(\la^2)}{1+2 \la 
\int_{-\infty}^\infty v^{q-1}(x) z(x) dx+ O(\la^2)}= \\
&=& J^{\max}_\ve+
2\la (\dpr{\cm v}{z} - J^{\max}_\ve \dpr{v^{q-1}}{z})+O(\la^2) 
\end{eqnarray*} 
Since $J_\ve\left(\f{v+\la z}{\|v+\la z\|_{L^q}}\right)\leq J^{\max}_\ve$, 
we conclude that for all $z\in C^\infty_0(a,b)$ 
\begin{equation}
\label{p:10}
\dpr{\cm v -J^{\max}_\ve v^{q-1} }{z}\leq 0.
\end{equation} 
Since there are no rectrictions on $z$ (other than compact support), it follows that \\  $\dpr{\cm v -J^{\max}_\ve v^{q-1} }{z}=0$ and hence  
\begin{equation}
\label{p:15}
\cm v -J^{\max}_\ve v^{q-1}=0,
\end{equation} 
in the $\cd'(a,b)$ sense. This is the Euler-Lagrange equation that we were looking for.  
Note that according to our derivation, it holds in the compact subsets of $\Om$. Note that from here, we obtain   $v=\left(\f{1}{J^{\max}_\ve} \cm v\right)^{\f{1}{q-1}}$, which is smooth. These could be iterated further to show that $v\in C^\infty$ 
on  open sets $(a,b)$, whenever $[a,b]\subset \Om$.  
\subsubsection{Euler-Lagrange over  non-trivial intervals of $\Om^c$} 
\label{sec:4.2.2}
Our goal will be to show that such \underline{non-trivial intervals do not exist}. To that end, we will first assume that they do exist and then, we will be able to derive an Euler-Lagrange equation on them, which will then lead to a contradiction. 
Take such an interval, say $[a_0, a_1]\subset \Om^c$, $0\leq a_0<a_1$. \\
\\
{\bf Claim:} $v=const.$ in $\cd'[a_0, a_1]$ sense, i.e. $v$ is constant on connected 
components\footnote{But it may be a different constant on a different component} of $\Om^c$. 
\begin{proof} 
To prove that, assume the opposite, namely that $v$ is not a constant in  $\cd'(a_0,a_1)$. Thus (since we know $v'\leq 0$), there exists a test function $0\leq \psi_0\in \cd(a_0,a_1)$, so that $\dpr{v'}{\psi_0}=-\alpha<0$. 
For some small $\ka\in (0,1)$,   the set $V_\ka=\{x\in (a_0, a_1): \psi_0(x)>\ka\}$ will be nonempty and open. Take an interval $(\tilde{a}_0, \tilde{a}_1)\subset V_\ka$, so that 
$\tilde{a}_1-\tilde{a}_0<1$. Take an arbitrary function 
$\psi\in \cd(\tilde{a}_0, \tilde{a}_1)$, so that $0\leq \psi(y)<\ka$ 
for $y\in (\tilde{a}_0, \tilde{a}_1)$. It follows that $\psi(y)<\ka<\psi_0(y)$ and hence 
$\dpr{v'}{\psi}\leq \dpr{v'}{\psi_0}=-\al<0$. Also 
$\int\psi(y) dy\leq \ka(\tilde{a}_1-\tilde{a}_0)<1$, whence 
\begin{equation}
\label{pl1}
\dpr{v'}{\psi}\leq -\al<-\al \int\psi(y) dy,
\end{equation}
for all $\psi\in \cd'(\tilde{a}_0, \tilde{a}_1)$, with $0\leq \psi<\ka$. One can now extend \eqref{pl1} to hold for all $0\leq \psi\in \cd'(\tilde{a}_0, \tilde{a}_1)$. Hence, $(\tilde{a}_0, \tilde{a}_1)\subset \Om$ in contradiction with    $(\tilde{a}_0, \tilde{a}_1)\subset (a_0, a_1)\subset \Om^c$. 
\end{proof}
Now that we have established that $v$ is a.e. constant on any non-trivial interval 
$(a_0, a_1)\subset \Om^c$, we will derive the Euler-Lagrange equation for \eqref{7} on it. 
We will consider first the case $a_0>0$, the other case will be considered separately. \\
{\bf Case I: $a_0>0$}\\
To fix the ideas, we consider first the case when the interval 
$[a_0,a_1]$ is isolated from the rest of $\Om^c$, that is, there exists $r>0$, so that 
$(a_0-r,a_0)\subset \Om$ and $(a_1, a_1+r)\subset \Om$. 

Fix  $\theta>0$ be so small that 
$(a_{0}-\theta, a_{0}), (a_{1}, a_{1}+\theta) \subset\Om$. Consider a test function $z\in C^\infty_0(a_{0}-\theta, a_{1}+\theta)$.  

Clearly, since $v=const$ on $(a_{0}, a_{1})$, {\it we need to require that the function $z'\leq 0$} on $(a_{0}, a_{1})$, in order for $v+\la z$ to be bell-shaped  function\footnote{so that it is acceptable entry for the minimization problem \eqref{7}} (recall $\la>0$). 
Denote  
\begin{equation}
\label{p:40}
G(x)=\cm v-J^{\max}_\ve v^{q-1}. 
\end{equation}
Since $v=const.$ a.e. on $(a_0,a_1)$,  $G$ is a continuous function on 
$(a_0, a_1)$. Following the approach of the previous section, we derive the equation \eqref{p:10}, which in the new notation says that  
\begin{equation}
\label{p:20}
\dpr{G}{z}\leq 0,
\end{equation}
for all $z$, so that $v+\la z$ is an admissible entry for \eqref{7}. We will show that  
 \eqref{p:20} implies 
\begin{eqnarray}
\label{p:25}
& & \int_{a_{0}}^{b} G(x) dx\leq 0\\
\label{p:27}
& & \int_{b}^{a_{1}} G(x) dx\geq 0.
\end{eqnarray}
for all $b \in (a_{0}, a_{1})$. Applying \eqref{p:25} for $b=a_1-$ and \eqref{p:27} for 
$b=a_0+$ and taking limits, yields $\int_{a_0}^{a_1} G(x) dx=0$. \\
\\
{\bf Claim:} From \eqref{p:25} and \eqref{p:27}, one can infer that  $G=0$ on $(a_{0}, a_{1})$. 
\begin{proof}
By the definition of $G$, the fact that $v=const.$ 
 and Lemma \ref{le:3}, we conclude that $G$ is non-increasing and continuous 
 function on $(a_{0}, a_{1})$. Now, we have several cases.  
 If there exists a $b\in (a_0, a_1)$, so that $ \int_{a_{0}}^{b} G(x) dx<0$ and 
 $\int_{b}^{a_{1}} G(x) dx>0$, then by continuity, there will be $b_0\in (a_0,b)$, so that $G(b_0)<0$ and $b_1\in (b,a_1)$, so that  $G(b_1)>0$, a contradiction with the fact that $G$ is non-increasing. Otherwise, for all $b\in (a_0, a_1)$, we have that either $\int_{a_0}^b G(x) dx=0$ or $\int_{b}^{a_{1}} G(x) dx=0$. But since  $\int_{a_0}^{a_1} G(x) dx=0$, 
 $$
 \int_{b}^{a_{1}} G(x) dx=\int_{a_0}^{a_1} G(x) dx - \int_{a_0}^b G(x) dx=-\int_{a_0}^b G(x) dx. 
 $$
Thus, $\int_{a_0}^b G(x) dx=0$ for all $b\in (a_0, a_1)$. Thus, for any  $b_0, b_1\in (a_0, a_1)$, we get $\int_{b_0}^{b_1} G(x) dx=\int_{a_0}^{b_1} G(x) dx- 
\int_{a_0}^{b_0} G(x) dx=0$ and hence $G(x)\equiv 0$ in $(a_0, a_1)$. 
 \end{proof}
  Thus, the Euler-Lagrange equation is in 
  the form $G(x)=0$ for $x\in (a_{0}, a_{1})$ or 
 \begin{equation}
 \label{pli}
\cm v-J^{\max}_\ve v^{q-1}=0, 
 \end{equation}
 if we can show \eqref{p:25} and \eqref{p:27}. 

We show only \eqref{p:25}, the proof of \eqref{p:27} is similar. 
Fix $b: a_{0}<b< a_{1}$. For all $0<\de<<1$, introduce an {\it even} test function, which for $x>0$ is given by 
$$
z_\de(x)=\left\{
\begin{array}{ll}
0 & 0<x<a_{0} - 2\de \\
1 & a_{0} - \de<x<b + \de\\
0 & x> b + 2 \de
\end{array}
\right.
$$
so that $z_\de\in C^\infty$, $z_\de$ is \underline{strictly increasing} in 
$(a_0 - 2\de,a_0 - \de)$  and $z_\de$ is \underline{strictly decreasing} in 
$(b + \de, b+2\de)$. \\

Note that  $v+\la z_\de$ is acceptable for the maximization problem \eqref{7} for\footnote{Here even though $z_\de$ is increasing in $(a_0 - 2\de,a_0 - \de)$, 
this is acceptable since $(a_0 - 2\de,a_0 - \de)\subset \Om$ and hence, we have no restrictions over the test functions, as long as $\la_\de<<1$} $0<\la<\la_\de$.  According to \eqref{p:20}, $\dpr{G}{z_\de}\leq 0$.  Hence 
\begin{eqnarray*}
& & \int_{a_0}^b G(x) dx=\limsup_{\de\to 0} 
\int_{a_{0} - 2\de}^{b + 2 \de} G(x) z_\de(x) dx  \leq  0,
\end{eqnarray*}
which is \eqref{p:25}. 

The general case, in which $[a_0, a_1]$ is not isolated from $\Om^c$ (i.e. there is no $r>0$, so that $(a_0-r, a_0) \cup (a_1, a_1+r)\subset \Om$) is treated in the same way. Indeed, this was needed only in the very last step, in the construction of the function $z_\de$. But clearly, one can carry out a similar construction of $z_\de$, if one has a sequence of intervals $I_{\de_j}=(c_{\de_j}-\de_j, c_{\de_j}) \subset \Om, c_{\de_j}<a_0$, so that $\lim c_{\de_j}=a_0$. Similarly, for the proof of \eqref{p:27}, one needs a sequence of intervals $J_{\de_j}=(d_{\de_j}, d_{\de_j}+\de_j) 
\subset \Om, d_{\de_j}>a_1$, so that $\lim d_{\de_j}=a_1$. 

Finally, it remains to observe that every non-trivial interval of $\Om^c$ is contained in 
$[a_0, a_1]\subseteq \Om^c$ with the property $dist(a_0, \Om)=0=dist(a_1, \Om)$ and hence, we can carry the constructions of $z_\de$ and hence the validity of 
\eqref{p:25} and \eqref{p:27} follows. We then derive \eqref{pli} on every such interval $(a_0, a_1)$.  \\
{\bf Case II: $a_0=0$}\\
This case is similar to the previous one. We again assume that there exists $r>0$, so that $(-a_1-r, -a_1)\cup (a_1, a_1+r)\subset \Om$, the general case being reduced to this one by arguments similar to those in the case $a_0>0$.  

Again,  $G$ is even, we have that $v=const.$
 on the interval $(-a_1,a_1)$ and $G$ is non-increasing in $(0,a_1)$. 
 We will show that for every $b \in (0, a_1)$, we have 
 \begin{eqnarray}
 \label{p:30}
 & &  \int_{b}^{a_1} G(x) dx\geq 0,    \\
 \label{p:35}
 & &  \int_{0}^{b} G(x) dx\leq 0
 \end{eqnarray}
 Let us first prove  that assuming \eqref{p:30} and \eqref{p:35}, 
 one must have $G=0$. Indeed, if 
 we apply \eqref{p:30} for $b\to 0+$ and \eqref{p:35} for $b\to a_1-$, we see that $\int_0^{a_1} G(x) dx=0$. 
 
 Again, if we assume that $\int_0^b G(x) dx=0$ for all $b\in (0,a_1)$, we again conclude by the 
  continuity of $G$ that $G(x)=0: 0\leq x\leq a_1$. If one has 
  $\int_0^{b_0} G(x) dx<0$ for some $b_0\in (0,a_1)$, we have that 
 $\int_{b_0}^{a_1} G(x) dx> 0$ (since $\int_{0}^{a_1} G(x) dx=0$) and hence, $G$  cannot be non-increasing function (as in Case I), a contradiction. Thus $G(x)=0$ in $(0,a_1)$. 
 
 Thus, it remains to show \eqref{p:30} and \eqref{p:35}. 
 Fix $b>0$ and let $0<\de<<1$. For \eqref{p:35}, 
 select  
$$
z_\de(x)=\left\{
\begin{array}{ll}
0 & x<-b-2\de \\
1 & -b- \de<x<b+\de\\
0 & x> b+ 2 \de
\end{array}
\right.
$$
and $z_\de$ is $C^\infty$, even and  increasing in $(-b-2\de, -b-\de)$ and decreasing in $(b+\de, b+2\de)$. Clearly, for $\la>0$, $v+\la z_\de$ is admissible for \eqref{7}, for some small $\la>0$. Thus, applying \eqref{p:20} for this $z_\de$ and passing to appropriate limits as $\de\to 0+$, we obtain, as above 
$$
2\int_0^b G(x) dx= \int_{-b}^b G(x) dx=\liminf_{\de\to 0+} \int G(x) z_\de(x) dx \leq 0. 
$$
For \eqref{p:30}, we construct $z_\de$ to be an even function, so that  
$$
z_\de(x)=\left\{
\begin{array}{ll}
0 & x<b-2\de \\
-1 & b- \de<x<a_1+\de\\
0 & x> a_1 + 2 \de
\end{array}
\right.
$$
Now, we require that $z_\de$ is decreasing in $(b-2\de, b-\de)$ and it is increasing from 
$(a_1+\de, a_1+2\de)$. Note that $z_\de$ is still acceptable as a perturbation - in the sense that $v+\la z_\de$ is non-increasing in $(0,\infty)$ for all $\la=\la(\de)>0$ small enough (for this, recall that  $(a_1+\de, a_1+2\de)\subset \Om$ and hence, there exists $\sigma(\de)$, so that $v'<-\sigma$ in $\cd'(a_1+\de, a_1+2\de)$ sense). 
We have again 
$$
-\int_{b}^{a_1} G(x) dx=\limsup_{\de\to 0+} \int G(x) z_\de(x)dx \leq 0. 
$$
Thus, we have established \eqref{p:30} and \eqref{p:35} and thus the Euler-Lagrange equation 
\begin{equation}
\label{p:67}
0=G(x)=\cm v- J^{\max}_\ve v^{q-1}.
\end{equation}
\subsubsection{The set $\Om^c$ consists of isolated points only} 
 \label{sec:4.2.3}  
We will now show that if an equation like \eqref{p:67} holds in a non-trivial interval, say $(a,b)$ and $v$ is a bell-shaped, locally integrable function, 
which is constant on $(a,b)$, then $v(x)=const.$ a.e. on $\rone$ (which would be a contradiction). 

Indeed, $\cm v$ is in fact a differentiable function on $(a,b)$, which is non-increasing on $(0,\infty)$, according to Lemma \ref{le:3}. Thus, taking a derivative of \eqref{p:67} (and taking into account that $v=const$ on $(a,b)$)  leads to 
\begin{equation}
\label{y}
0=(\cm v)'(x)=\int_{x}^{x+1} v(y) dy- \int_{x-1}^{x} v(y) dy.
\end{equation}
for all $x\in (a,b)$. 

If $b>1$, we see that since $v$ is non-increasing in $(0,\infty)$ (by Lemma \ref{le:p}), it follows that $x\to \int_{x}^{x+1} v(y) dy$ is non-increasing and continuous. By \eqref{y}, it follows that $\int_{x}^{x+1} v(y) dy=const.$ for $x\in (a,b)$. Hence,  we must have $v=const.$ a.e. in every  interval in the form $(x-1, x+1)$ (for all $x\in (a,b)$). By iterating this argument in all $(0,\infty)$, a contradiction. If $(a,b)\subset (0,1)$, we can again argue as in Lemma \ref{le:3} to establish that again 
$v=const.$  in $(0,\infty)$. Thus, we cannot have non-trivial intervals 
$(a,b)\subset \Om^c$ and hence $\Om^c$ consists of isolated points only. 

\subsubsection{The Euler-Lagrange equation on $\rone$} 
Before deriving the Euler-Lagrange equation for the maximizer $v$, 
let us recapitulate what we have shown so far for $v$. We managed to show that 
 $\Om$ is a dense open set, so that $\Om^c$ consists of isolated points only.  
 Finally, on $\Om$, $v$ is a continuous function, and the equation \eqref{p:15} holds on every interval  $(a,b)\subset \Om$.  
 
 We will now show that \eqref{p:15} holds for almost all $x\in \Om^c$. First of all, recall that \\ 
 $G=\cm v-J_{\ve}^{\max}  v^{q-1}$ is locally integrable function (as a sum of an $L^q$ function and $L^{\f{q}{q-1}}$ functions) and hence, almost all points are Lebesgue points for it. 
 
 Let $x_0>0, x_0\in \Om^c$, so that $x_0$ is a Lebesgue point for $G$. We have shown that $x_0$ is an isolated point of $\Om^c$, which implies the existence of intervals inside $\Om$, which approximate $x_0$. That is, there are 
 $(a_j-\de_j, a_j+\de_j)\subset \Om^+$, $(b_j-\de_j, b_j+ \de_j) \subset \Om^+$, 
 so that $a_j+\de_j< x_0<b_j-\de_j$ and  $\lim_j a_j=\lim_j b_j=x_0$. In addition, we can clearly select these intervals to be very short, namely we require 
 $\lim_j \f{\de_j}{b_j-a_j}=0$. Construct now a sequence of even $C^\infty$ test functions, given by  
 $$
z_j(x)=\left\{
\begin{array}{ll}
0 & 0<x<a_j-\de_j/2  \\
1 & a_j <x<b_j\\
0 & x> b_j + \de_j/2
\end{array}
\right.
 $$
where $z_j$ is strictly increasing in $(a_j-\de_j/2, a_j)\subset \Om$ and strictly decreasing in $(b_j, b_j+\de_j/2)$. We have already shown that  functions of the form 
$v\pm \la z_j$ will be non-increasing in $(0,\infty)$  and it will otherwise satisfy all the restrictions of the optimization problem \eqref{7}, provided $0<\la<<1$. 
Thus, accoridng to \eqref{p:10}, we have $\dpr{G}{z_j}\leq 0$ and $-\dpr{G}{z_j}\leq 0$ and hence  
$$
 \int_{a_j}^{b_j} G(x) dx+ \int_{a_j-\de_j/2}^{a_j} G(x) z_j(x) dx + 
 \int_{b_j}^{b_j+\de_j/2} G(x) z_j(x) dx=  \dpr{G}{z_j}= 0
$$
Dividing both sides by $b_j-a_j$, and taking $\lim$  as $j\to \infty$ (noting that $x_0$ is a Lebesgue point), we get 
\begin{equation}
\label{z}
G(x_0)+
\lim \f{\de_j}{(b_j-a_j)} [\f{1}{\de_j} \int_{a_j-\de_j/2}^{a_j} G(x) z_j(x) dx+
\f{1}{\de_j} \int_{b_j}^{b_j+\de_j/2} G(x)z_j(x)dx]=0. 
\end{equation}
Thus, we need to show that the limit above exists and it is equal to zero. 
To that end, note that $\|z_j\|_\infty\leq 1$ and 
\begin{eqnarray*}
|\f{1}{\de_j} \int_{a_j-\de_j/2}^{a_j} G(x) z_j(x) dx+
\f{1}{\de_j} \int_{b_j}^{b_j+\de_j/2} G(x)z_j(x)dx| &\leq & \|G\|_{L^\infty_x}\leq 
(\|v\|_{L^\infty}+J_{\ve}^{\max}\|v\|_{L^\infty}^{q-1}) \\
&\leq & (v(0)+ J_{\ve}^{\max} v(0)^{q-1}),
\end{eqnarray*}
since $v$ is non-increasing a.e. in $(0, \infty)$. 
The last inequality, combined with $\lim_j \f{\de_j}{b_j-a_j}=0$ 
shows that \eqref{z} implies $G(x_0)=0$. Hence, for all Lebesgue points of $G$, we have $G(x)=0$. Thus, 
\begin{equation}
\label{lk}
\cm v- J_{\ve}^{\max} v^{q-1}=0 \ \ a.e. 
\end{equation}
\subsection{Taking a limit as $\ve\to 0$: Constructing a solution to \eqref{p:1}} 
\label{sec:4.3} 
From \eqref{lk}, 
\begin{equation}
\label{p:45}
\cm v_\ve- J^{\max}_\ve v_\ve^{q-1}=0,
\end{equation}
which is satisfied $a.e.$ and also in 
$\cd'(0, \ve^{-1}-1)$ sense. From it, we learn that $v_\ve$ is 
$H^4(0, \ve^{-1}-1)$ and consequently, by iterating this argument,  
$H^\infty(0, \ve^{-1}-1)=\bigcap_{s=1}^\infty H^s(0, \ve^{-1}-1)$  function. Recall also that by construction $\|v_\ve\|_{L^q(\rone)}=1$ and there exists $\ve_0>0$,  so that $\inf_{\ve<\ve_0} J^{\max}_\ve>0$, see \eqref{p:55}. We will now take several consecutive subsequences of $\ve\to 0$, in order to ensure that the limit satisfies \eqref{p:1}. 

First, take $\ve_n\to 0+$, so that $\lim_n J^{\max}_{\ve_n}=\limsup_{\ve\to 0} J^{\max}_\ve:=J_0>0$. Second, out of this constructed sequence $\ve_n$, take a subsequence, say $\ve_{n_k}$, so that $v_{\ve_{n_k}}\rightharpoonup v$ in a weak $L^q$ sense, for some $v\geq 0, v \in L^q$. This is possible, by the sequential compactness of the unit ball in the weak $L^q$ topology\footnote{But again, this argument, so far,  does not guarantee that $v\neq 0$}. By the uniqueness of weak limits (by eventually taking further subsequence), we also get $v^{q-1}_{\ve_{n_k}} \rightharpoonup_k v^{q-1}$ in weak $L^{\f{q}{q-1}}$ sense.  We also have $\cm v_{\ve_{n_k}} \rightharpoonup \cm v$ in the weak $L^q$ topology, since for every test function $\psi\in C^\infty_0$, we have by the self-adjointness of $\cm$ and $\cm:L^r\to L^r, 1\leq r\leq \infty$,  
$$
\dpr{\cm v_{\ve_{n_k}}}{\psi}= \dpr{  v_{\ve_{n_k}}}{\cm \psi}\to_k \dpr{v}{\cm \psi}=\dpr{\cm v}{\psi}.
$$ 
Thirdly,  we show that the limiting function $v$ is non-zero. To that end, it will suffice to establish   that 
\begin{equation}
\label{p:60} 
\liminf_{\ve\to 0+} \|v_\ve\|_{L^{q}(-\ve^{-1}, \ve^{-1})}>0.
\end{equation}
Assuming that \eqref{p:60} is false, we will reach a contradiction. Indeed, let $\de_j\to 0+$ be a sequence  so that 
$\lim_j \|v_{\de_j}\|_{L^{q}(-\de_j^{-1}, \de_j^{-1})}=0$. Thus,   
\begin{eqnarray*}
J^{\max}_{\de_j}=J_{\de_j}(v_{\de_j}) &\leq &  \int_{-\de_j^{-1} }^{\de_j^{-1} }|Q v_{\de_j}(x)|^2 dx= 
\int_{-\de_j^{-1} }^{\de_j^{-1} } |\int_{x-1/2}^{x+1/2} v_{\de_j}(y) dy|^2dx. 
\end{eqnarray*}
We now use a refined version of the Gagliardo-Nirenberg estimate that we have used before.   
\begin{eqnarray*}
\int_{-\de_j^{-1} }^{\de_j^{-1} } |\int_{x-1/2}^{x+1/2} v_{\de_j}(y) dy|^2dx &=&  \int_{-\de_j^{-1} }^{\de_j^{-1} }|\int_{x-1/2}^{x+1/2} v_{\de_j}(y)
\chi_{(-\de_j^{-1}-1, \de_j^{-1}+1)} dy|^2 dx\leq  \\
&\leq & \|Q[v_{\de_j}\chi_{(-\de_j^{-1}-1, \de_j^{-1}+1)}]\|_{L^2}^2 
\end{eqnarray*}
and hence
\begin{eqnarray*}
\|Q[v_{\de_j}\chi_{(-\de_j^{-1}-1, \de_j^{-1}+1)}]\|_{L^2} \leq  
\|Q[v_{\de_j}\chi_{(-\de_j^{-1}-1, \de_j^{-1}+1)}]\|_{\dot{W}^{1,q}}^{1/q-1/2}
\|Q [v_{\de_j}\chi_{(-\de_j^{-1}-1, \de_j^{-1}+1)}]]\|_{L^q}^{3/2-1/q}. 
\end{eqnarray*}
While a simple differentiation shows that on one hand, 
$$
\|Q[v_{\de_j}\chi_{(-\de_j^{-1}-1, \de_j^{-1}+1)}]\|_{\dot{W}^{1,q}}\leq 2 \|v_{\de_j}\|_{L^q}=2,
$$
we also have  by Cauchy-Schwartz 
 \begin{eqnarray*}
\|Q [v_{\de_j} \chi_{(-\de_j^{-1}-1, \de_j^{-1}+1)}]\|_{L^q}^q &\leq &  \int_{-\infty}^{\infty} (\int_{x-1/2}^{x+1/2} v_{\de_j}(y)\chi_{(-\de_j^{-1}-1, \de_j^{-1}+1)}(y) dy)^q dx \leq \\
& \leq & 
\int_{-\de_j^{-1}-1}^{\de_j^{-1}+1} v_{\de_j}^q (y) dy \leq 
2 \|v_{\de_j}\|_{L^q(-\de_j^{-1}, \de_j^{-1})}^q \to 0. 
\end{eqnarray*}
The last inequality here follows by the fact that $v_{\de_j}$ is 
non-increasing in $(0,\infty)$ and therefore  
$\int_{\de_j^{-1}}^{\de_j^{-1}+1}v_{\de_j}^q (y) dy \leq \int_{0}^{\de_j^{-1}}v_{\de_j}^q (y) dy$, if $\de_j^{-1}\geq 1$, which we have assumed anyway.  
Thus, we will have proved that 
$$
\liminf_j J^{\max}_{\de_j}\leq 0,
$$
which is in contradiction with \eqref{p:55}. Thus, we have established \eqref{p:60}. 

We are now ready to take a limit as $\ve\to 0$ in \eqref{p:45}. Indeed, take \eqref{p:45} for $\ve=\ve_{n_k}, k=1, 2, \ldots$. Fix a test function $\psi\in C^\infty_0$. There exists $k_0$, so that for $k\geq k_\psi$, $supp \ \psi\subset (0,\ve_{n_k}^{-1})$.  Thus, for $k\geq k_\psi$, 
we get\footnote{by testing \eqref{p:45}, which is valid on the support of $\psi$} 
$$
\dpr{\cm v_{\ve_{n_k}}}{\psi} - J^{\max}_{\ve_{n_k}}  \dpr{v_{\ve_{n_k}}^{q-1}}{\psi}=0.
$$
Take a limit as $k\to \infty$. By our constructions, we have that 
$J^{\max}_{\ve_{n_k}}  \dpr{v_{\ve_{n_k}}^{q-1}}{\psi}\to J_0 \dpr{v^{q-1}}{\psi}$, since $J^{\max}_{\ve_{n_k}}\to J_0$, $\dpr{v_{\ve_{n_k}}^{q-1}}{\psi}\to \dpr{v^{q-1}}{\psi}$ by the weak $L^{\f{q}{q-1}}$ convergence. 

We also have 
$\dpr{\cm v_{\ve_{n_k}}}{\psi}=\dpr{v_{\ve_{n_k}}}{\cm  \psi}\to_k \dpr{v}{\cm  \psi}=
\dpr{\cm v}{\psi}$. 
Thus, 
we have established the desired identity
\begin{equation}
\label{fin}
\cm v - J_0 v^{q-1} =0 
\end{equation}
valid for all $x>0$. By the symmetry, it is also valid for $x<0$. It is clear that $v$ is now infinitely smooth\footnote{starting with 
$v\in L^q(\rone)$, it is easy to conclude that $\cm v$ is smooth, which in turn implies that $v\in C^1(\rone)$ etc.}  function on $(0, \infty)$.  Recall though, that for Theorem \ref{theo:1} we needed to solve $\cm w - w^{q-1}=0$. One can easily construct $w$, based on the solution $v$ of \eqref{fin}. More precisely, if we take 
$w:=J_0^{-\f{1}{2-q}} v$,  then $w$ will satisfy \eqref{p:2}. Theorem \ref{theo:1} is proved. 
\section{An alternative proof of the Friesecke-Wattis theorem}
We quickly indicate how our ideas can be turned into a new proof of the Friesecke-Wattis theorem. Specifically, as we saw in the previous section, it is clear that if we are just interested in the existence of traveling waves for \eqref{p:1} (but not in bell-shaped solutions), 
it is a good idea to consider  following constrained maximization problem (compare to \eqref{7}) 
\begin{equation}
\label{77}
\left| 
\begin{array}{c}
J_\ve(v)=\int_{-\ve^{-1}}^{\ve^{-1}}|Q v(x)|^2 dx   \to \max \\
\textup{subject to} \ \ v\geq 0; \  I(v)=\int_{\rone} v^q(x)dx =1.
\end{array}
\right.
\end{equation}
First off, the arguments in Section \ref{sec:3.1} apply unchanged (by just skipping the bell-shapedness of $v$) to prove that \eqref{77} has a maximizer, say $v$. 

Following the argument of Section \ref{sec:3.2} and more specifically, the following  identities 
\begin{eqnarray*}
& & \|v+\la z\|_{L^q}^q=\|v\|_{L^q}^q + \la q \dpr{v^{q-1}}{z}+O(\la^2) \\
& & J_\ve(v+\la z)=  J_\ve(v)+2\la \dpr{Q^2 v}{z}+O(\la^2), \\
& & J_\ve\left(\f{v+\la z}{\|v+\la z\|_{L^q}}\right) =   J^{\max}_\ve+
2\la (\dpr{\cm v}{z} - J^{\max}_\ve \dpr{v^{q-1}}{z})+O(\la^2)\leq J_\ve(v),
\end{eqnarray*}
which were established there, it follows that 
$$
\dpr{\cm v-J^{\max}_\ve v^{q-1}}{z}\leq 0
$$
for all test functions\footnote{Note that here, there is no restriction whatsoever  on the increasing/decreasing character of $z$, due to the nature of \eqref{77}}   $z$. It follows that $v$ satisfies 
$$
\cm v - J^{\max}_\ve v^{q-1}=0.
$$
This of course produces a family $v_\ve: \ supp\  v_\ve\subset (-\ve^{-1}-1, \ve^{-1}+1)$, which easily can be shown to converge\footnote{following the methods of Section \ref{sec:4.3}}  (after an eventual subsequence) to a $v: \ supp\  v\subset \rone$, which solves 
$$
\cm v - J_0 v^{q-1}=0.
$$
Setting $w:=J_0^{-\f{1}{2-q}} v$ again provides a solution to $w^{1/p}= \cm w$ as is required by \eqref{p:2}.


\begin{thebibliography}{100}


\bibitem{flagor} S. Flach and A. Gorbach, Phys. Rep. {\bf 467}, 1 (2008).

\bibitem{leder} F. Lederer, G.I. Stegeman, D.N. Christodoulides,
G. Assanto, M. Segev and Y. Silberberg,
Phys. Rep. {\bf 463}, 1 (2008).

\bibitem{markus} O. Morsch and M. Oberthaler, 
Rev. Mod. Phys. {\bf 78}, 179 (2006).

\bibitem{peyrard} M. Peyrard,
Nonlinearity {\bf 17}, R1 (2004).

\bibitem{nesterenko1} V.~F. Nesterenko, {\it Dynamics of Heterogeneous Materials} (Springer-Verlag, New York, NY, 2001).

\bibitem{sen08} S. Sen, J. Hong, J. Bang, E. Avalos, and R. Doney, Phys. Rep. {\bf 462}, 21 (2008).

\bibitem{nesterenko2} C. Daraio, V.~F. Nesterenko, E.~B. Herbold, and S. Jin, Phys. Rev. E. {\bf 73}, 026610 (2006).

\bibitem{coste97} C. Coste, E. Falcon, and S. Fauve, Phys. Rev. E. {\bf 56}, 6104 (1997)

\bibitem{dar06} C. Daraio, V.~F. Nesterenko, E.~B. Herbold, and S. Jin, Phys. Rev. Lett. {\bf 96}, 058002 (2006).

\bibitem{hong05} J. Hong, Phys. Rev. Lett. {\bf 94}, 108001 (2005).

\bibitem{fernando} F. Fraternali, M.~A. Porter, and C. Daraio, Mech. Adv. Mat. Struct., in press (arXiv:0802.1451).

\bibitem{doney06} R. Doney and S. Sen, Phys. Rev. Lett. {\bf 97}, 155502 (2006).

\bibitem{dev08} D. Khatri, C. Daraio, and P. Rizzo, SPIE {\bf 6934}, 69340U (2008).


\bibitem{dar05} C. Daraio, V.~F. Nesterenko, E.~B. Herbold, and S. Jin, Phys. Rev. E {\bf 72}, 016603 (2005).

\bibitem{dar05b} V.~F. Nesterenko, C. Daraio, E.~B. Herbold, and S. Jin, Phys. Rev. Lett. {\bf 95}, 158702 (2005).

\bibitem{chiaraus} M.A. Porter, C. Daraio, E.B. Herbold,
  I. Szelengowicz
and P.G. Kevrekidis,
Phys. Rev. E {\bf 77}, 015601 (2008); M.A. Porter, C. Daraio, 
I. Szelengowicz, E.B. Herbold 
and P.G. Kevrekidis, Phys. D {\bf 238}, 666 (2009).

\bibitem{chiarag} N. Boechler, G. Theocharis, S. Job, 
P. G. Kevrekidis, M.A. Porter, and C. Daraio
Phys. Rev. Lett. {\bf 104}, 244302 (2010);
 G. Theocharis, N. Boechler, P.G. Kevrekidis, S. Job, M.A. Porter, 
and C. Daraio
Phys. Rev. E {\bf 82}, 056604 (2010).

\bibitem{FPU0} E. Fermi, J. Pasta, S. Ulam,
Los Alamos National Laboratory report LA-1940 (1955).

\bibitem{FPU1} D.K. Campbell, P. Rosenau and G.M. Zaslavsky, Chaos
{\bf 15}, 015101 (2005).

\bibitem{pegogf1} G. Friesecke and R.L. Pego,
Nonlinearity {\bf 12}, 1601 (1999); {\it ibid.} {\bf 15}, 1343 (2002);
{\it ibid.} {\bf 17}, 207 (2004); {\it ibid.} {\bf 17}, 2229 (2004).

\bibitem{flach1} M. Eleftheriou, B. Dey and G.P. Tsironis,
Phys. Rev. E {\bf 62}, 7540 (2000); 
B. Dey, M. Eleftheriou, S. Flach and G.P. Tsironis,
Phys. Rev. E {\bf 65}, 017601 (2002).

\bibitem{pego} J.M. English and R.L. Pego,
Proceedings of the AMS {\bf 133}, 1763 (2005).

\bibitem{pikovsky} K. Ahnert and A. Pikovsky,
Phys. Rev. E {\bf 79}, 026209 (2009).

\bibitem{mackay} R.S. MacKay, Phys. Lett. A {\bf 251}, 191 (1999).

\bibitem{fw} G. Friesecke and J.A.D. Wattis, Comm. Math. Phys.
{\bf 161}, 391 (1994).

\bibitem{chatter} A. Chatterjee,
Phys. Rev. E {\bf 59}, 5912 (1998).




\bibitem{FLL} J. Fr\"ohlich, E. H. Lieb, M. Loss, \emph{Stability of Coulomb systems with magnetic fields. I. The one-electron atom.} {\em  Comm. Math. Phys.} {\bf  104} 
  (1986),  no. 2, p. 251--270.
  
\bibitem{KM} P. Karageorgis, P.J. McKenna, 
\emph{Existence of ground states for fourth order wave equations}, preprint, available at  http://arxiv.org/abs/1004.2775

\bibitem{Lieb}  E. H. Lieb, \emph{On the lowest eigenvalue of the Laplacian for the intersection of two domains.} {\em  Invent. Math.} {\bf  74}  (1983),  no. 3, p. 441--448.

 \bibitem{Rudin} W. Rudin, Functional Analysis, Second Edition, 1991. 


 
  
 
 
 



 

\end{thebibliography}
\end{document}